\documentclass[journal]{IEEEtran}
\usepackage{caption}
\usepackage{diagbox}
\usepackage{verbatim}
\usepackage[usenames]{color}
\usepackage{amssymb}
\usepackage{graphicx}
\usepackage{amscd}
\usepackage{enumitem}

\usepackage[colorlinks=true,
linkcolor=webgreen,
filecolor=webbrown,
citecolor=webgreen]{hyperref}

\definecolor{webgreen}{rgb}{0,.5,0}
\definecolor{webbrown}{rgb}{.6,0,0}

\usepackage{color}
\usepackage{fullpage}
\usepackage{float}

\usepackage{graphics,amsmath,amssymb}
\usepackage{amsthm}
\usepackage{amsfonts}
\usepackage{latexsym}
\usepackage{epsf}

\DeclareMathOperator{\U}{\mathcal{U}}
\DeclareMathOperator{\lso}{lso}
\DeclareMathOperator{\so}{so}
\newcommand{\M}{M}

\usepackage{breakurl}

\begin{document}

\theoremstyle{plain}
\newtheorem{theorem}{Theorem}
\newtheorem{corollary}[theorem]{Corollary}
\newtheorem{lemma}[theorem]{Lemma}
\newtheorem{proposition}[theorem]{Proposition}

\theoremstyle{definition}
\newtheorem{definition}[theorem]{Definition}
\newtheorem{example}[theorem]{Example}
\newtheorem{conjecture}[theorem]{Conjecture}

\theoremstyle{remark}
\newtheorem{remark}[theorem]{Remark}

\author{
Daniel Gabric
}
%\thanks{Copyright (c) 2017 IEEE. Personal use of this material is permitted.  However, permission to use this material for any other purposes must be obtained from the IEEE by sending a request to pubs-permissions@ieee.org.}
%\IEEEpubid{\begin{minipage}[t]{\textwidth}\ \\[10pt]
%        \centering{Copyright (c) 2017 IEEE. Personal use of this material is permitted.  However, %permission to use this material for any other purposes must be obtained from the IEEE by sending a %request to pubs-permissions@ieee.org.}
%\end{minipage}} 

\title{Mutual Borders and Overlaps}

\maketitle

\begin{abstract}
A word is said to be \emph{bordered} if it contains a non-empty  proper prefix that is also a suffix. We can naturally extend this definition to pairs of non-empty words. A pair of words $(u,v)$ is said to be \emph{mutually bordered} if there exists a word that is a non-empty proper prefix of $u$ and suffix of $v$, and there exists a word that is a non-empty proper suffix of $u$ and  prefix of $v$. In other words, $(u,v)$ is mutually bordered if $u$ overlaps $v$ and $v$ overlaps $u$. We give a recurrence for the number of mutually bordered pairs of words. Furthermore, we show that, asymptotically, there are $c\cdot k^{2n}$ mutually bordered words of length-$n$ over a $k$-letter alphabet, where $c$ is a constant. Finally, we show that the expected shortest overlap between pairs of words is bounded above by a constant.
\end{abstract}

\section{Introduction}

Let $\Sigma_k$ denote the alphabet $\{0,1,\ldots, k-1\}$. Let $\Sigma_k^n$ denote the set of all length-$n$ words over the alphabet $\Sigma_k$. The length of a word $w$ is denoted by $|w|$. A \emph{border} of a word $w \in \Sigma_k^*$ is a non-empty word that is both a proper prefix and suffix of $w$. A word with a border is said to be \emph{bordered}. Otherwise, the word is said to be \emph{unbordered}. For example, the French word {\tt entente} is bordered, and has two borders, namely {\tt ente} and {\tt e}. Unbordered words are also referred to as \emph{bifix-free} words, where a \emph{bifix} is just a border. We choose to use \emph{border} to denote a non-empty prefix that is also a suffix since it was used first~\cite{Silberger:1971, Ehrenfeucht&Silberger:1979}, and it is more widely used in the literature.

Let $\U_n^k$ denote the set of length-$n$ unbordered words over a $k$-letter alphabet. It is well known~\cite{Nielsen:1973} that the sequence $u_n=|\U_n^k|$ is defined by the recurrence  
\[u_n =
\begin{cases} 
      1, & \text{if }n=0;\\
      ku_{n-1} - u_{n/2}, & \text{if $n>0$ is even;}\\
      ku_{n-1}, & \text{if $n$ is odd.}
   \end{cases}
\]
In the same paper by Nielsen, he showed that the limit $\lim\limits_{n\to \infty} u_n/k^n$ exists. In particular, he showed that for $k=2$ there are $(c + o(1))\cdot 2^n$ unbordered binary words, where $c \approx 0.267786$.   Also see~\cite{Holub&Shallit:2016}. 

The notion of a word being unbordered can naturally be generalized to pairs of words. Our goal is to prove results similar to Nielsen's for these kinds of pairs of words. The concept of a word being unbordered can also be generalized to larger collections of words. But counting the number of such tuples (or sets) of words can get very complicated because of all of the different ``interacting" words. Thus we focus on the simplest generalization to multiple words, that is, to pairs of words.

Let $u$ and $v$ be words of length $m$ and $n$, respectively. Let $w$ be a non-empty word. In this paper we write $(u,v)$ to refer to an ordered pair of words. We say that $(u,v)$ has a \emph{right-border} if $u$ has a non-empty proper suffix that is a proper prefix of $v$. If $w$ is a suffix of $u$ and prefix of $v$ then $w$ is said to be a \emph{right-border} of $(u,v)$.  Analogously, we say that $(u,v)$ has a \emph{left-border} if $u$ has a non-empty proper prefix that is a proper suffix of $v$. If $w$ is a prefix of $u$ and suffix of $v$ then $w$ is said to be a \emph{left-border} of $(u,v)$.\footnote{We could have defined left-borders and right-borders to refer to ordinary non-empty prefixes and suffixes without specifying they be proper. But since a border is defined as a non-empty proper prefix and suffix of a word, we decided to keep the definition analogous.}

A pair of words $(u,v)$ is said to be \emph{mutually bordered} if $(u,v)$ has both a right-border and a left-border. If $(u,v)$ has neither a right-border nor a left-border, then $(u,v)$ is said to be \emph{mutually unbordered}. The pair $(u,v)$ is said to be \emph{right-bordered} if $(u,v)$ has a right-border but not a left-border. Similarly $(u,v)$ is said to be \emph{left-bordered} if $(u,v)$ has a left-border but not a right-border.

\begin{example}
The pair of English words ({\tt delivered}, {\tt redeliver}) is mutually bordered. The word {\tt red} is a right-border of the pair and {\tt deliver} is a left-border of the pair.

The pair of English words ({\tt mail}, {\tt box}) is mutually unbordered since it has no right-border or left-border. 

The pair of English words ({\tt overlap}, {\tt lapse}) is right-bordered. The word {\tt lap} is a right-border of the pair. \end{example}

Mutually unbordered words have previously arisen in digital communications as a generalization of a method of frame synchronization~\cite{BajicStojanovic:2004,WijngaardenWillink:2000}. The goal of frame synchronization is to let the receiver of some piece of data know where the boundaries of the frames in the data are (i.e., both the sender and receiver are on the same page). Typically this is done by inserting a specially chosen word periodically into the data stream as a kind of delimiter. 

In 2000 van Wijngaarden and Willink~\cite{WijngaardenWillink:2000} proposed a new method of frame synchronization where a set of different words are interleaved into the data stream periodically instead of appearing as a contiguous subword. An important part of frame synchronization is the detection of the periodically inserted word. In 2004 Bajic and Stojanovic~\cite{BajicStojanovic:2004} calculated statistical quantities related to the detection of distributed sequences in random data. For more work on mutually unbordered words, also called ``cross-bifix-free words" or ``non-overlapping words", see~\cite{Bajic:2014, Barcucci:2021, Bernini:2017,Bilotta:2017,BilottaPergola:2012,Blackburn:2015,Chee:2013,StefanovicBajic:2012}. We choose not to use the terminology ``cross-bifix-free words" or ``non-overlapping words" since we require a specific name for pairs of words $(u,v)$ with the property that $(u,v)$ has either a right-border or a left-border but not both.

\begin{itemize}
    \item Let $\M_k(m,n)$ denote the number of mutually bordered pairs of words $(u,v)$.
    \item Let $R_k(m,n)$ denote the number of right-bordered pairs of words $(u,v)$. 
    \item Let $U_k(m,n)$ denote the number of mutually unbordered pairs of words $(u,v)$. 
\end{itemize}
See Tables~\ref{table:M},~\ref{table:R}, and~\ref{table:U} for sample values of $\M_k(m,n)$, $R_k(m,n)$, and $U_k(m,n)$ for $m$, $n$ where $1\leq m,n\leq 8$.

In this paper, we are primarily concerned with pairs of equal length words (i.e., the case where $m=n$). So we define $M_k(n) = M_k(n,n)$, and we define $R_k(n)$ and $U_k(n)$ similarly (see Table~\ref{table:MRU} for some sample values). The main results of this paper are Corollary~\ref{cor:fourthirds}, Theorem~\ref{theorem:enumeration}, Theorem~\ref{theorem:limiting}, and Theorem~\ref{thm:expected}. In Corollary~\ref{cor:fourthirds} we bound the sum of the length of the shortest left-border and right-border of a pair of words. In Theorem~\ref{theorem:enumeration} we give recurrences for $M_k(n)$, $R_k(n)$, and $U_k(n)$. Then in Theorem~\ref{theorem:limiting} we prove that the limit $\lim\limits_{n\to \infty}\M_k(n)/k^{2n}$ exists. Finally, in Theorem~\ref{thm:expected} we show that the expected shortest right-border and left-border of a pair of equal-length words is $O(1)$.

\newpage
\begin{table}[H]
\caption{Some values of $M_2(m,n)$ for $m$, $n$ where $1 \leq m,n \leq 8$.}
\centering

\begin{tabular}{|c|cccccccc|}
\hline
\backslashbox{$m$}{$n$}& 1  & 2 & 3 & 4 & 5 & 6 & 7 & 8 \\
\hline
1 & 0 &  0  & 0   & 0  &  0  &  0  &   0   &  0\\
2 & 0 & 4 &  8 &  16 &  32 &  64 &  128  & 256\\
3 & 0 & 8 & 26 &  50  &100 & 200  & 400  & 800\\
4 & 0 & 16 & 50  &124 & 242 & 484  & 968 & 1936\\
5 & 0 & 32& 100&  242 & 524& 1036  &2070 & 4142\\
6 & 0 & 64& 200 & 484& 1036& 2154  &4280 & 8554\\
7 & 0 & 128& 400&  968 &2070& 4280 & 8706& 17354\\
8 & 0 & 256& 800 &1936& 4142& 8554& 17354& 34996\\
\hline
\end{tabular}

\captionsetup{justification=centering}
\label{table:M}
\end{table}

\begin{table}[H]
\caption{Some values of $R_2(m,n)$ for $m$, $n$ where $1 \leq m,n \leq 8$.}
\centering
\begin{tabular}{|c|cccccccc|}
\hline
\backslashbox{$m$}{$n$}& 1  & 2 & 3 & 4 & 5 & 6 & 7 & 8  \\
\hline
1 & 0 &  0  & 0   & 0  &  0  &  0  &   0   &  0\\
2 & 0 & 4 &  8 &  16 &  32 &  64 &  128  & 256\\
3 & 0 & 8 & 14 &  30  &60 & 120  & 240  & 480\\
4 & 0 & 16 & 30 & 52 & 110 & 220 & 440 &  880\\
5 & 0 & 32 & 60& 110 & 204 & 420 & 842 & 1682\\
6 & 0 & 64& 120 &220 & 420 & 806& 1640 & 3286\\
7 & 0 & 128& 240& 440 & 842 &1640& 3214 & 6486\\
8 & 0 & 256& 480& 880& 1682& 3286 &6486& 12844\\
\hline
\end{tabular}
\captionsetup{justification=centering}
\label{table:R}
\end{table}

\begin{table}[H]
\caption{Some values of $U_2(m,n)$ for $m$, $n$ where $1 \leq m,n \leq 8$.}
\centering
\begin{tabular}{|c|cccccccc|}
\hline
\backslashbox{$m$}{$n$} & 1  & 2 & 3 & 4 & 5 & 6 & 7 & 8  \\
\hline
1 & 4 &  8 & 16 & 32 & 64 & 128 & 256 & 512  \\
2 & 8 &  4 &  8 & 16 & 32  & 64 & 128 & 256 \\
3 & 16 & 8 & 10 & 18 & 36  & 72 & 144 & 288  \\
4 & 32 & 16 & 18 & 28 & 50 & 100 & 200 & 400  \\
5 & 64 & 32 & 36 & 50 & 92 & 172 & 342 & 686  \\
6 & 128 & 64 & 72 & 100 & 172 & 330 & 632 & 1258 \\
7 & 256 & 128 & 144 & 200 & 342 &  632 & 1250 & 2442  \\
8 & 512 & 256 & 288 & 400 & 686 & 1258 & 2442 & 4852  \\
\hline
\end{tabular}
\captionsetup{justification=centering}
\label{table:U}
\end{table}

\begin{table}[H]
\caption{Some values of $M_2(n)$, $R_2(n)$, and $U_2(n)$ for $n$ where $1 \leq n \leq 15$.}
\centering
\begin{tabular}{c|c|c|c}
$n$ & $M_2(n)$  & $R_2(n)$ & $U_2(n)$  \\
\hline
1 & 0 & 0 & 4 \\
2 & 4 & 4 & 4 \\
3 & 26 & 14 & 10 \\
4 & 124 & 52 & 28 \\
5 & 524 & 204 & 92 \\
6 & 2154 & 806 & 330 \\
7 & 8706 & 3214 & 1250 \\
8 & 34996 & 12844 & 4852 \\
9 & 140290 & 51366 & 19122 \\
10 & 561724 & 205492 & 75868 \\
11 & 2247892 & 822108 & 302196 \\
12 & 8993414 & 3288858 & 1206086 \\
13 & 35976928 & 13156624 & 4818688\\
14 & 143913546 & 52629590 & 19262730\\
15 & 575664422 & 210525818 & 77025766
\end{tabular}
\captionsetup{justification=centering}
\label{table:MRU}
\end{table}

\section{Number of mutually bordered pairs}
In this section we enumerate the number of mutually bordered pairs of words $\M_k(n)$.\begin{itemize}
    \item Let $\so(u,v)$ denote the shortest right-border of $(u,v)$, and let $\so(u,v)=\epsilon$ if $(u,v)$ does not have a right-border. By definition we have that $\so(v,u)$ is the shortest left-border of $(v,u)$, and $\so(v,u)=\epsilon$ if $(u,v)$ does not have a left-border. \item Let $\lso(v,u)$ be the length of the shortest right-border of $(u,v)$, and let $\lso(u,v) = 0$ if $(u,v)$ does not have a right-border. Again we have that by definition $\lso(v,u)$ is the length of the shortest left-border of $(u,v)$, and $\lso(v,u)=0$ if $(u,v)$ does not have a left-border.
\end{itemize}
\begin{example}
The pair of binary words $(u,v) = (1000101, 0110001)$ has one right-border and two left-borders. The word $01$ is a right-border of the pair. The words $1$ and $10001$ are left-borders of the pair. The shortest right-border is $01$ and is of length $2$. The shortest left-border is $1$ and is of length $1$. So $\so(u, v) = 01$, $\lso(u,v) = 2$, $\so(v,u)= 1$, and $\lso(v,u) = 1$.
\end{example}

Since the concept of a pair of words being mutually bordered is similar to the concept of a single word being bordered, it is natural to assume that the insights used to count bordered words might be useful to also count pairs of mutually bordered words. Let $u$ be a bordered word. Let $v$ be the shortest border of $u$. The key ideas used to count length-$n$ bordered words is that for bordered words $u$, the shortest border $v$ is unbordered, and $|v| \leq n/2$. Combining these facts we get the following formula for the number of length-$n$ bordered words over a $k$-letter alphabet:
\[\sum_{i=1}^{\lfloor n/2\rfloor} {u_i}\cdot k^{n-2i}.\]

The basic idea we will use to count mutually bordered pairs of words $(u,v)$ is to, analogously to bounding the length of the shortest border, put a bound on $\lso(u,v) + \lso(v,u)$ (see Corollary~\ref{cor:fourthirds}). Then further classify all pairs $(u,v)$ into two groups based on $\lso(u,v)+\lso(v,u)$. If $\lso(u,v)+\lso(v,u)$ is `small', then we can easily count the number of such pairs merely by using the number of unbordered words (see Lemma~\ref{lemma:smallwords}). If $\lso(u,v)+\lso(v,u)$ is `large', then the pair $(u,v)$ has a certain structure we can exploit to count them (see Lemma~\ref{lemma:bigwords}).

\begin{lemma}\label{lemma:short}
Let $n\geq 1$. Let $u,v\in \Sigma_k^n$. Let $w$ be a non-empty word that is both a proper suffix of $u$ and a proper prefix of $v$. Then $w=\so(u,v)$ iff $w$ is unbordered.
\end{lemma}
\begin{proof}

\noindent We prove an equivalent proposition, namely that $w\neq \so(u,v)$ iff $w$ is bordered.

\noindent
$w\neq \so(u,v)$ $\iff$  There exists a non-empty word $x$ such that $|x| < |w|$ and $x=\so(u,v)$. $\iff$ $x = \so(w,w)$ $\iff$ $w = xs = tx$ for some $s,t\in \Sigma_k^+$ $\iff$  $w$ is bordered.
\begin{comment}
$\Longleftarrow:$ Suppose $w$ is bordered. Then there exists a non-empty word $y$ that is a proper prefix and suffix of $w$. But such a word would also be a suffix of $u$, and a prefix of $v$. Therefore we have $w\neq \so(u,v)$.
\end{comment}
\end{proof}

\begin{lemma}\label{lemma:smallwords}
Let $n\geq 1$. Let $u$ and $v$ be length-$n$ words. Let $i=\lso(u,v)$, and $j=\lso(v,u)$. Then $i+j \leq n$ iff $u = xsy$ and $v = ytx$ for some words $s,t\in \Sigma_k^*,x\in \U_j^k,$ and $y\in \U_i^k$.
\end{lemma}
\begin{proof}
\noindent

$\Longrightarrow:$ Let $y=\so(u,v)$ and $x = \so(v,u)$. Let $j = \lso(u,v)$ and $i = \lso(v,u)$. Suppose $i+j \leq n$. By definition there exist words $w,z,\alpha,\beta\in \Sigma_k^*$ such that $u = wy, v=yz, v=\alpha x,$ and $u = x\beta$. But since $i+j = |x| + |y| \leq n$, we have that $x$ and $y$ do not overlap. Thus there exist words $s,t\in \Sigma_k^{n-i-j}$ such that $u = xsy$ and $v = ytx$. By Lemma~\ref{lemma:short}, we have that $x$ and $y$ must be unbordered. Therefore $x \in \U_j^k$ and $y\in \U_i^k$

$\Longleftarrow:$ Follows from the definition of $u$ and $v$.
\end{proof}

\begin{lemma}\label{lemma:bigwords}
Let $n\geq 1$. Let $u$ and $v$ be length-$n$ words. Let $i=\lso(u,v)$, and $j=\lso(v,u)$. Then $i+j > n$ iff 
\begin{enumerate}
    \item [a)] $n+1\leq i+j\leq \frac{4}{3}n$, and
    \item [b)] there exist distinct words $x,y\in \Sigma_k^{i+j - n},$ and $s,t\in \Sigma_k^*$ such that $u$ is of the form $xsytx$ and $v$ is of the form $ytxsy$ where $(x,y)$ is mutually unbordered, and both $xsy$ and $ytx$ are unbordered with $\so(u,v)=ytx$ and $\so(v,u)=xsy$.
\end{enumerate}
\end{lemma}
\begin{proof}
\noindent

$\Longrightarrow:$
Since $i+j > n$ we have that $(\so(u,v),\so(v,u))$ has a right-border and a left-border. Let $x$ be the length-$(i+j-n)$ suffix of  $\so(u,v)$. We can now write $u = r\alpha x = xw\beta$ and $v = \alpha xw$ for some $r,w,\alpha,\beta\in \Sigma_k^*$ where $|\alpha x| = i$ and $|xw| = j$. Clearly $x$ is a prefix and suffix of $u$. Let $y$ be the length-$(i+j-n)$ suffix of $\so(v,u)$. By a similar argument as above, one can show that $y$ is a prefix and suffix of $v$. If $(x,y)$ has a right-border or a left-border, then $(u,v)$ has a right-border of length $<i$ or a left-border of length $<j$. So $(x,y)$ must be mutually unbordered and $u$ must be of the form $xsytx$ and $v$ must be of the form $ytxsy$ for some $s,t\in \Sigma_k^*$ where $|ytx| = i$ and $|xsy| = j$. Since $(x,y)$ is mutually unbordered, the words $u$ and $v$ can only be of this form if $2|x| + |y| = 2|y|+ |x| \leq n$. Since $|x| = |y| = (i+j-n)$, we have that $3(i+j-n) \leq n \implies i+j \leq \frac{4}{3}n$. Now by Lemma~\ref{lemma:short}, we have that $ytx$ and $xsy$ are unbordered. Since both $xsy$ and $ytx$ are unbordered and $|x|=|y|$, we have that $y$ and $x$ must be distinct.

$\Longleftarrow:$ We have that $i+j>n$ by assumption.
\end{proof}

Perhaps the most peculiar and interesting aspect to Lemma~\ref{lemma:bigwords} is the fact that for length-$n$ words $u$ and $v$, the sum of $\lso(u,v)$ and $\lso(v,u)$ is bounded by a number between $n$ and $2n$. This fact is outlined in Corollary~\ref{cor:fourthirds}.
\begin{corollary}
Let $n\geq 1.$ Let $u$ and $v$ be length-$n$ words. Then $\lso(u,v) + \lso(v,u) \leq \frac{4}{3}n.$
\label{cor:fourthirds}
\end{corollary}

In Example~\ref{example:max} we illustrate pairs of words $(u,v)$ that reach the bound $\frac{4}{3}n$ bound.

\begin{example}
\label{example:max}
The following three pairs of words illustrate the upper bound $\lfloor\frac{4}{3}n\rfloor$ from Lemma~\ref{lemma:bigwords} and Corollary~\ref{cor:fourthirds}. We give examples for all lengths of words by giving examples for each equivalence class modulo $3$.

\indent
For $n= 3m$, we have
\[(u,v) = (0^m1^m0^m,1^m0^m1^m).\]
\indent For $n= 3m+1$, we have
\[(u,v) = (0^m1^{m+1}0^m,1^{m+1}0^m1^m).\]
\indent For $n= 3m+2$, we have
\[(u,v) = (0^m1^{m+2}0^m,1^{m+2}0^m1^{m}).\]
\end{example}

Lemma~\ref{lemma:bigwords} shows that pairs of length-$n$ non-empty words $(u,v)$ where $\lso(u,v) + \lso(v,u)$ is `large' ($>n$) exhibit a particular structure. Namely $\so(u,v)$ is unbordered and $\so(u,v)$ begins and ends with a mutually unbordered pair of words. The same is true for $\so(v,u)$ as well.  So we need an expression for the number of unbordered words whose prefix and suffix of a certain length form a pair of mutually unbordered words.

Let $t\leq n$ be a positive integer. Let $u$ and $v$ be length-$t$ words such that $(u,v)$ is mutually unbordered. Let $G_{u,v}(n)$ denote the number of length-$n$ unbordered words that have $u$ as a prefix, $v$ as a suffix (and vice versa).

\begin{lemma}\label{lemma:unbordered}
Let $n\geq t\geq 1$. Let $u$ and $v$ be length-$t$ words such that $(u,v)$ is mutually unbordered and $u\neq v$. Then the number of unbordered words that have $u$ as a prefix and $v$ as a suffix is 
\[ G_{u,v}(n) = \begin{cases} 
      0, & \text{if $n < 2t$;} \\
      k^{n-2t} - \sum_{i=2t}^{\lfloor n/2\rfloor} G_{u,v}(i)k^{n-2i}, & \text{if $n\geq 2t$.}
   \end{cases}
\]
\end{lemma}
\begin{proof}
If $n<2t$ then  $G_{u,v}(n) = 0$, since $(u,v)$ is mutually unbordered. Suppose $n\geq 2t$. Then the number of unbordered words of length $n$ having $u$ as a prefix and $v$ as a suffix is equal to the number of bordered words that contain $u$ as a prefix and $v$ as a suffix,  subtracted from the total number of words that contain $u$ as a prefix and $v$ as a suffix. Let $w$ be a word of length $n$ such that $u$ is a prefix of $w$ and $v$ is a suffix of $w$. Then $w$ is bordered if and only if its shortest border is of length $j$ where $2t\leq j \leq \lfloor n/2\rfloor$. This is because $(u,v)$ is mutually unbordered and words that have a border of length $>n/2$ must also have a border of length $\leq n/2$. Also observe that the shortest border of $w$ must itself be unbordered and have $u$ as a prefix and $v$ as a suffix. So the total number of words of the form $w$ as described above is $\sum_{i=2t}^{\lfloor n/2\rfloor} G_{u,v}(i)k^{n-2i}$. Therefore, for $n\geq 2t$ we have 
\[G_{u,v}(n) = k^{n-2t} - \sum_{i=2t}^{\lfloor n/2\rfloor} G_{u,v}(i)k^{n-2i}.\]
\end{proof}

From Lemma~\ref{lemma:unbordered} we see that $G_{u,v}$ is independent of $u$ and $v$, but dependent on the length of $|u| = |v|$.  Therefore, for $u$, $v$ words of length $t$, let $G_{u,v}(n) = G_t(n)$.

Also notice that for $n,t\geq 1$ we have,
\begin{equation}
G_t(n) \leq k^{n-2t}.
    \nonumber
\end{equation}
This follows from the fact that $G$ counts a subset of length-$n$ words that have a fixed length-$t$ prefix and suffix.

Finally, we are ready to present recurrences for $M_k(n)$, $R_k(n)$, and $U_k(n).$
\begin{theorem}\label{theorem:enumeration}
The number $\M_k(n)$ of mutually bordered pairs of words of equal length satisfies
\begin{align}
\M_k(n)= &\sum_{i=1}^{n-1}\sum_{j=1}^{n-i}u_iu_jk^{2n-2(i+j)} + \label{eqn:mutualborder} \\ 
&\sum_{i=1}^{\lfloor n/3\rfloor}(U_k(i) - u_i)\sum_{j=2i}^{n-i}G_i(j)G_i(n-j+i),    \nonumber
\end{align}
where $\M_k(1)=0$. Additionally we have 
\begin{equation}
    R_k(n) =\bigg(\sum_{i=1}^{n-1}k^{2n-2i}u_i\bigg) - \M_k(n),\label{eqn:mutualhalfborder}
\end{equation}
and
\begin{equation}
    U_k(n) + 2R_k(n) + M_k(n) = k^{2n}. \label{eqn:mutualunborder}
\end{equation}
\end{theorem}
\begin{proof}
Let $n\geq 1$, and let $u$ and $v$ be words of length $n$.

\bigskip

\noindent Proof of Eq.~(\ref{eqn:mutualunborder}): Let $(u,v)$ be a pair of length-$n$ words. Exactly one of the following must be true about $(u,v)$:
\begin{enumerate}[label=(\alph*)]
    \item $(u,v)$ has a right-border and a left-border (i.e., $(u,v)$ is mutually bordered),
    \item $(u,v)$ has a right-border but not a left-border (i.e., $(u,v)$ is right-bordered),
    \item $(u,v)$ does not have a right-border but has a left-border (i.e., $(u,v)$ is left-bordered),
    \item $(u,v)$ does not have a right-border or a left-border (i.e., $(u,v)$ is mutually unbordered).
\end{enumerate}
Clearly the number of right-bordered pairs of words is the same as the number of pairs of left-bordered pairs of words. From these facts we conclude that $U_k(n) + 2R_k(n) + M_k(n) = k^{2n}$.

\bigskip

\noindent Proof of Eq.~(\ref{eqn:mutualhalfborder}): The number of right-bordered pairs of words is equal to the total number of pairs $(u,v)$ that have a right-border subtracted from the total number of mutually bordered pairs of words. So $R_k(n) = ( \sum_{i=1}^{n-1}k^{2n-2i}u_i) - \M_k(n)$.

\bigskip

\noindent Proof of Eq.~(\ref{eqn:mutualborder}): Clearly $\M(1) = 0$ since words of length $1$ cannot have left-borders or right-borders. Let $i=\lso(u,v)$ and $j = \lso(v,u)$. By Lemma~\ref{lemma:smallwords} we have that $i+j \leq n$ iff $u = xsy$ and $v = ytx$ for some words $s,t\in \Sigma_k^*,x\in \U_j^k,$ and $y\in \U_i^k$. We can count the number of pairs of such words using the number of unbordered words, $\sum_{i=1}^{n-1}\sum_{j=1}^{n-i}u_iu_jk^{2n-2(i+j)}$. 

By Lemma~\ref{lemma:bigwords} we have that $i+j > n$ iff $n+1\leq i+j\leq \frac{4}{3}n$ and there exist words $x,y\in \Sigma_k^{i+j - n},$ and $s,t\in \Sigma_k^*$ such that $u$ is of the form $xsytx$ and $v$ is of the form $ytxsy$ where $(x,y)$ is mutually unbordered, both $xsy$ and $ytx$ are unbordered with $|ytx| = i$ and $|xsy| = j$, and $x\neq y$. The fact that $n+1\leq i+j \leq \frac{4}{3}n$ and $i$, $j$, $n$ are integers imply that $1\leq |x| = |y| = (i+j-n) \leq \lfloor n/3\rfloor$. Let $p=(i+j-n)$. Since $(x,y)$ is mutually unbordered and both $xsy$ and $ytx$ are unbordered, we have that $x\neq y$. Suppose that we in fact have $w = x = y$ for some $w \in \Sigma_k^p$. The only such $w$ must be unbordered, since $(x,y)$ is unbordered. Therefore, the number of mutually unbordered pairs $(x,y)$ with $x\neq y$ is $U_k(p) - u_p$.

By Lemma~\ref{lemma:unbordered} we know that for $(x,y)$ fixed, the number of unbordered words of the form $ytx$ with $(x,y)$ mutually unbordered and $x\neq y$ is $G_{x,y}(i)=G_{p}(i)$. Similarly the number of unbordered words of the form $xsy$ is $G_{x,y}(n-i+p)=G_{p}(n-i+p)$. We also  have that $i \geq 2p$ and $i\leq n- p$ since $(x,y)$ is mutually unbordered. For $(x,y)$ fixed, the total number of pairs of words of the form $(xsy,ytx)$ is $\sum_{l=2p}^{n-p} G_{p}(l)G_{p}(n-l+p)$. So the number of words of the form $xsytx$ and $ytxsy$ as described above is equal to the number of pairs of words $(xsy,ytx)$ with $(x,y)$ mutually unbordered, $xsy$ and $ytx$ unbordered, and $x\neq y$. Since $1\leq p = (i+j-n) \leq \lfloor n/3\rfloor$ we have that this is equal to $\sum_{p=1}^{\lfloor n/3\rfloor}(U_k(p) - u_p)\sum_{l=2p}^{n-p}G_p(l)G_p(n-l+p)$.

Putting it all together, we have 
\begin{align}
\M_k(n)= &\sum_{i=1}^{n-1}\sum_{j=1}^{n-i}u_iu_jk^{2n-2(i+j)} + \nonumber \\ 
&\sum_{i=1}^{\lfloor n/3\rfloor}(U_k(i) - u_i)\sum_{j=2i}^{n-i}G_i(j)G_i(n-j+i).    \nonumber
\end{align}
\end{proof}

\section{Limiting values}
In this section we show that the limit $L_k=\lim\limits_{n\to \infty} \M_k(n)/k^{2n}$ exists. 

\begin{theorem}\label{theorem:limiting}
The following limit exists: \[L_k=\lim_{n\to \infty} \frac{\M_k(n)}{k^{2n}}.\] Furthermore, we have that
\[\bigg(\sum_{i=1}^{n}u_ik^{-2i}\bigg)^2 \leq L_k \leq \bigg(\bigg(\sum_{i=1}^{n}u_ik^{-2i}\bigg) + \frac{k^{-n}}{k-1}\bigg)^2.\]
\end{theorem}
\begin{proof}
From the recurrence for $\M_k(n)$, we see that there are two main terms. The first term is $\sum_{i=1}^{n-1}\sum_{j=1}^{n-i}u_iu_jk^{2n-2(i+j)}$, which counts all pairs of words $(u,v)$ where $\lso(u,v) +\lso(v,u)$ is `small'. The second term is $\sum_{i=1}^{\lfloor n/3\rfloor}(U_k(i)- u_i)\sum_{j=2i}^{n-i}G_i(j)G_i(n-j+i)$, which counts all pairs of words $(u,v)$ where $\lso(u,v)+\lso(v,u)$ is `large'. Now from Lemma~\ref{lemma:bigwords}, we know that the only pairs $(u,v)$ where $\lso(u,v)+\lso(v,u)$ is `large' are pairs of words of the form $(xsytx,ytxsy)$. There are at most $k^n$ choices for $xsytx$, and $ytxsy$ can be recreated from $xsytx$ by knowing the starting positions of $s$, $y$, and $t$. Therefore there must be $o(k^{2n})$ pairs of such  words. So in the limit $L_k$, this term goes to $0$, and thus $L_k=\lim\limits_{n\to \infty} \sum_{i=1}^{n-1}\sum_{j=1}^{n-i}u_iu_jk^{-2(i+j)}$. Therefore we have that the limit $L_k$ exists if and only if $\lim_{n\to \infty} \sum_{i=1}^{n-1}\sum_{j=1}^{n-i}u_iu_jk^{-2(i+j)}$ converges.

Consider the infinite double series \[L_k' =\lim_{n\to \infty} \sum_{i=1}^{n-1}\sum_{j=1}^{n-1}u_iu_jk^{-2(i+j)}.\] We can factor out $u_ik^{-2i}$ out of the nested series and split up the nested sum to get  
\begin{align}
L_k' &= \lim_{n\to \infty}\sum_{i=1}^{n-1}u_ik^{-2i}\sum_{j=1}^{n-1}u_jk^{-2j} \nonumber \\
     &= \lim_{n\to \infty}\sum_{i=1}^{n-1}u_ik^{-2i}\bigg(\sum_{j=1}^{n-i}u_jk^{-2j} + \sum_{j=n-i+1}^{n-1}u_jk^{-2j}\bigg).\nonumber
\end{align}
For each $i$, we have  $\lim\limits_{n\to \infty} \sum_{j=n-i+1}^{n-1}u_jk^{-2j}=0$, and thus, $L_k = L_k'$.

So using the fact that $L_k=L_k'$, we have

\[L_k =\bigg(\sum_{i=1}^{\infty}u_ik^{-2i}\bigg)^2.\]

Thus the limit $L_k$ exists if and only if the series $\sum_{i=1}^{\infty}u_ik^{-2i}$ converges. Since $u_i \leq k^i$ we have that $u_ik^{-2i} \leq k^{-i}$. Therefore by direct comparison we have that $\sum_{i=1}^{\infty}u_ik^{-2i}$ converges, and so the limit $L_k$ exists. Since $ u_ik^{-2i} \leq k^{-i}$ we have that $\sum_{i=m}^\infty u_ik^{-2i} \leq \sum_{i=m}^\infty k^{-i}$. Using this fact along with the fact that $\sum_{i=1}^n u_ik^{-2i} \leq \sum_{i=1}^\infty u_ik^{-2i}$, we get the following inequalities,

\begin{align}
\sum_{i=1}^{n}u_ik^{-2i} \leq \sum_{i=1}^{\infty}u_ik^{-2i} &\leq \bigg(\sum_{i=1}^{n}u_ik^{-2i}\bigg) + \sum_{i=n+1}^\infty k^{-i}\nonumber \\
&=\bigg(\sum_{i=1}^{n}u_ik^{-2i}\bigg) + \frac{k^{-n}}{k-1}.\nonumber
\end{align}
Now we have bounds for our limit,
\begin{align}
    \bigg(\sum_{i=1}^{n}u_ik^{-2i}\bigg)^2 \leq L_k&=\bigg(\sum_{i=1}^{\infty}u_ik^{-2i}\bigg)^2\nonumber \\
       &\leq \bigg(\bigg(\sum_{i=1}^{n}u_ik^{-2i}\bigg) + \frac{k^{-n}}{k-1}\bigg)^2\nonumber
\end{align}

\end{proof}

\begin{corollary}
The following limit exists:
\[\lim_{n\to \infty} \frac{R_k(n)}{k^{2n}}.\]
\end{corollary}
\begin{corollary}
The following limit exists:
\[\lim_{n\to \infty} \frac{U_k(n)}{k^{2n}}.\]\label{cor:limUnb}
\end{corollary}

Table~\ref{table:first} shows the behaviour of the functions $\M_k(n)$, $R_k(n)$, and $U_k(n)$ as $k$ increases. Interestingly, there are more mutually bordered pairs than not when $k=2$, but when $k$ increases the number of mutually unbordered pairs of words dominates.

\begin{table}[h]
\centering
\caption{Limits of recurrences as $k$ increases}
{\renewcommand{\arraystretch}{1.1}
\begin{tabular}{c|c|c|c}
$k$ & $\lim\limits_{n\to \infty} \frac{\M_k(n)}{k^{2n}}$ & $\lim\limits_{n\to \infty} \frac{R_k(n)}{k^{2n}}$ & $\lim\limits_{n\to \infty} \frac{U_k(n)}{k^{2n}}$ \\
\hline
2 & 0.536 & 0.196 & 0.072\\
3 & 0.196 & 0.247 & 0.310\\
4 & 0.098 & 0.215 & 0.473\\
5 & 0.058 & 0.182 & 0.578\\
10 & 0.012 & 0.098 & 0.792\\
100 & 0.000 & 0.010 & 0.980
\end{tabular}
}
\captionsetup{justification=centering}
\label{table:first}
\end{table}
\section{Expected shortest right-border}

In this section we compute expected value of $\lso(u,v)$ and $\lso(v,u)$ for length-$n$ words $u$ and $v$. Additionally, we show that the expected value tends to a constant.

Let $S_k(i,n)$ denote the number of pairs of length-$n$ words $(u,v)$ over a $k$-letter alphabet such that $\lso(u,v) = i$.

\begin{proposition}
Let $n,k,i\geq 1$. Then $S_k(i,n)=u_ik^{2(n-i)}$.
\end{proposition}
\begin{proof}
Follows directly from Lemma~\ref{lemma:short}.
\end{proof}

\begin{theorem}\label{thm:expected}
Let $n,k\geq 1$. Let $u$ and $v$ be length-$n$ words over a $k$-letter alphabet. Then the expected value of $\lso(u,v)$ is $O(1)$. 
\end{theorem}
\begin{proof}

\begin{align}
   \lim_{n\to\infty} \sum_{i=0}^{n}i\cdot \Pr[X=i] &= \lim_{n\to\infty}\frac{1}{k^{2n}}\sum_{i=1}^{n-1}i\cdot S_k(i,n)\nonumber \\
    &= \sum_{i=1}^{\infty}i\cdot u_i k^{-2i} . \nonumber
\end{align}

Since $u_i \leq k^i$, we have that $i\cdot u_ik^{-2i} \leq i\cdot k^{-i}$. Therefore by direct comparison, the series $\sum_{i=1}^{\infty}i\cdot u_i k^{-2i}$ converges. 

We clearly have that \[\sum_{i=1}^{n}i\cdot u_i k^{-2i} \leq \sum_{i=1}^{\infty}i\cdot u_i k^{-2i}.\]

Since $i\cdot u_ik^{-2i} \leq i\cdot k^{-i}$ we have that $\sum_{i=m}^\infty i\cdot u_ik^{-2i} \leq \sum_{i=m}^\infty i\cdot k^{-i}$. Using this we get the upper bound
\begin{align}
\sum_{i=1}^{\infty}i\cdot u_i k^{-2i}&\leq \bigg(\sum_{i=1}^{n}i\cdot u_i k^{-2i}\bigg)+ \sum_{i=n+1}^{\infty}i\cdot k^{-i}\nonumber \\ 
&=\bigg(\sum_{i=1}^{n}i\cdot u_i k^{-2i}\bigg)+ k^{-n}\frac{k(n+1)-n}{(1-k)^2}.\nonumber    
\end{align}
\end{proof}

Table~\ref{table:second} shows the behaviour the expected shortest right-border/left-border as $k$ increases. Interestingly, $k=2$ is the only value of $k$ for which the expected length of the shortest left-border and right-border is greater than $1$. For all other $k>2$ we have that the expected shortest left-border and right-border are less than one symbol in length. 

\begin{table}[h]
\centering
\caption{Asymptotic expected value of $\lso(u,v)$ and $\lso(v,u)$}
{\renewcommand{\arraystretch}{1.1}
\begin{tabular}{c|c}
$k$ & $\sum\limits_{i=1}^{\infty}i\cdot u_i k^{-2i}$  \\
\hline
2 & 1.156 \\
3 & 0.605\\
4 & 0.395 \\
5 & 0.290 \\
10 & 0.121 \\
100 & 0.010
\end{tabular}
}
\captionsetup{justification=centering}
\label{table:second}
\end{table}

\section{Open problems}
We conclude by posing some open problems:
\begin{itemize}
    \item How many pairs of length-$n$ words $(u,v)$ have a largest right-border/left-border of length $i$?
    \item What is the expected length of the longest right-border/left-border of a pair of words?
    \item Find recurrences for $\M_k(m,n)$, $R_k(m,n)$, and $U_k(m,n)$.
    \item Generalize to arbitrary tuples or sets of size $\geq 3$. Two obvious generalizations possible:
        \begin{enumerate}
            \item All consecutive words in a tuple are either mutually unbordered, or mutually bordered.
            \item All pairs of words in a set are either mutually unbordered (i.e., cross-bifix-free sets of words), or mutually bordered.
        \end{enumerate}
    \item The term \emph{Gray code} is used to describe an exhaustive listing of a set of combinatorial objects where successive terms differ by some small, well-defined amount. Gray codes are named after Frank Gray, who discovered a simple method of listing all $2^n$ binary words where successive words differ in exactly one position. Gray codes for cross-bifix-free sets already exist~\cite{Bernini&Bilotta:2017,Bernini:2014}. Can one generate a Gray code for mutually unbordered and bordered pairs of length-$n$ words where successive terms differ by a small amount if their Hamming distance is bounded by some constant $C\geq 1$? How small can $C$ get? 
    \item The concept of mutual borderedness and unborderedness can be extended to two dimensions~\cite{Barcucci:2017, Barcucci:2018} where words are two-dimensional matrices with entries taken from a finite alphabet. Similar to the study of mutually bordered and unbordered pairs of words in this paper, can one do the same for the set of all mutually bordered and unbordered pairs of $p\times q$-matrices?
\end{itemize}

\section{Acknowledgments}
The author would like to thank Jeffrey Shallit for suggesting the problem, and for valuable discussions on earlier drafts of this paper.
\bibliographystyle{unsrt}

\bibliography{abbrevs,main}

\newcommand{\noopsort}[1]{} \newcommand{\singleletter}[1]{#1}
\begin{thebibliography}{10}

\bibitem{Silberger:1971}
D.~M. Silberger.
\newblock Borders and roots of a word.
\newblock {\em Portugal. Math}, 30:191--199, 1971.

\bibitem{Ehrenfeucht&Silberger:1979}
A.~Ehrenfeucht and D.~M. Silberger.
\newblock Periodicity and unbordered segments of words.
\newblock {\em Discrete Math.}, 26:101--109, 1979.

\bibitem{Nielsen:1973}
P.~T. Nielsen.
\newblock A note on bifix-free sequences.
\newblock {\em IEEE Trans. Inform. Theory}, IT-19:704--706, 1973.

\bibitem{Holub&Shallit:2016}
S.~Holub and J.~Shallit.
\newblock Periods and borders of random words.
\newblock In Nicolas Ollinger and Heribert Vollmer, editors, {\em 33rd
  Symposium on Theoretical Aspects of Computer Science (STACS 2016)}, Leibniz
  International Proceedings in Informatics, pages 44:1--44:10. Schloss
  Dagstuhl---Leibniz Center for Informatics, 2016.

\bibitem{BajicStojanovic:2004}
D.~{Bajic} and J.~{Stojanovic}.
\newblock Distributed sequences and search process.
\newblock In {\em 2004 IEEE International Conference on Communications (IEEE
  Cat. No.04CH37577)}, volume~1, pages 514--518 Vol.1, 2004.

\bibitem{WijngaardenWillink:2000}
A.~J. {de Lind van Wijngaarden} and T.~J. {Willink}.
\newblock Frame synchronization using distributed sequences.
\newblock {\em IEEE Trans. Commun.}, 48(12):2127--2138, 2000.

\bibitem{Bajic:2014}
D.~Bajic and T.~Loncar-Turukalo.
\newblock A simple suboptimal construction of cross-bifix-free codes.
\newblock {\em Cryptography and Communications}, 6:27--37, 03 2014.

\bibitem{Barcucci:2021}
E.~Barcucci, A.~Bernini, and R.~Pinzani.
\newblock A strong non-overlapping {D}yck code.
\newblock In Nelma Moreira and Rog{\'e}rio Reis, editors, {\em Developments in
  Language Theory}, pages 43--53, Cham, 2021. Springer International
  Publishing.

\bibitem{Bernini:2017}
A.~Bernini, S.~Bilotta, R.~Pinzani, and V.~Vajnovszki.
\newblock A gray code for cross-bifix-free sets.
\newblock {\em Mathematical Structures in Computer Science}, 27(2):184–196,
  2017.

\bibitem{Bilotta:2017}
S.~{Bilotta}.
\newblock Variable-length non-overlapping codes.
\newblock {\em IEEE Trans. Inform. Theory}, 63(10):6530--6537, 2017.

\bibitem{BilottaPergola:2012}
S.~{Bilotta}, E.~{Pergola}, and R.~{Pinzani}.
\newblock A new approach to cross-bifix-free sets.
\newblock {\em IEEE Trans. Inform. Theory}, 58(6):4058--4063, June 2012.

\bibitem{Blackburn:2015}
S.~R. {Blackburn}.
\newblock Non-overlapping codes.
\newblock {\em IEEE Trans. Inform. Theory}, 61(9):4890--4894, 2015.

\bibitem{Chee:2013}
Y.~M. {Chee}, H.~M. {Kiah}, P.~{Purkayastha}, and C.~{Wang}.
\newblock Cross-bifix-free codes within a constant factor of optimality.
\newblock {\em IEEE Trans. Inform. Theory}, 59(7):4668--4674, July 2013.

\bibitem{StefanovicBajic:2012}
C.~{Stefanovic} and D.~{Bajic}.
\newblock On the search for a sequence from a predefined set of sequences in
  random and framed data streams.
\newblock {\em IEEE Trans. Commun.}, 60(1):189--197, 2012.

\bibitem{Bernini&Bilotta:2017}
A.~Bernini, S.~Bilotta, R.~Pinzani, and V.~Vajnovszki.
\newblock A {G}ray code for cross-bifix-free sets.
\newblock {\em Mathematical Structures in Computer Science}, 27(2):184–196,
  2017.

\bibitem{Bernini:2014}
A.~Bernini, S.~Bilotta, R.~Pinzani, A.~Sabri, and V.~Vajnovszki.
\newblock Prefix partitioned gray codes for particular cross-bifix-free sets.
\newblock {\em Cryptography and Communications}, 6(4):359--369, Dec 2014.

\bibitem{Barcucci:2017}
E.~Barcucci, A.~Bernini, S.~Bilotta, and R.~Pinzani.
\newblock Cross-bifix-free sets in two dimensions.
\newblock {\em Theoretical Computer Science}, 664:29--38, 2017.

\bibitem{Barcucci:2018}
E.~Barcucci, A.~Bernini, S.~Bilotta, and R.~Pinzani.
\newblock A 2{D} non-overlapping code over a $q$-ary alphabet.
\newblock {\em Cryptography and Communications}, 10(4):667--683, July 2018.

\end{thebibliography}

%\begin{IEEEbiographynophoto}{Daniel Gabric}
%received the B.Comp. (Hons.) and M.Sc. degrees in Computer Science from the University of Guelph in 2016 and 2018, respectively. He is currently a fourth year computer science Ph.D. student at the University of Waterloo. His research interests are in combinatorics on words, combinatorial algorithms, formal languages, and automata theory.
%\end{IEEEbiographynophoto}

\end{document}